\documentclass[letterpaper, 10 pt, conference]{ieeeconf}

\IEEEoverridecommandlockouts
\usepackage[utf8]{inputenc}
\usepackage{amsmath,amssymb}
\usepackage{graphicx}
\graphicspath{{figures/}}
\usepackage{cite}
\usepackage{algorithm}
\usepackage{algpseudocode}
\makeatletter
\newcommand\fs@ruledtop{\fs@ruled
  \def\@fs@pre{\kern5pt\hrule height.8pt depth0pt \kern2pt}}
\makeatother
\floatstyle{ruledtop}
\restylefloat{algorithm}
\usepackage[hidelinks]{hyperref}

\DeclareMathOperator{\tr}{tr}

\makeatletter
\def\th@plain{\itshape}
\def\@begintheorem#1#2{\trivlist
  \item[\hskip\labelsep\bfseries #1\ #2.]}
\def\@opargbegintheorem#1#2#3{\trivlist
  \item[\hskip\labelsep\bfseries #1\ #2\ (#3).]}
\makeatother

\newtheorem{theorem}{Theorem}
\newtheorem{lemma}[theorem]{Lemma}

\newtheorem{corollary}[theorem]{Corollary}
\newtheorem{assumption}{Assumption}

\title{\LARGE \bf Policy Improvement From Private Information For Decentralized LQG Teams}

\author{Hung Nguyen$^{1}$
and Andreas A. Malikopoulos$^{2}$,
\textit{Senior Member, IEEE}%
\thanks{This research was supported in part by NSF under Grants
CNS-2401007, CMMI-2348381, IIS-2415478, and in part by MathWorks.}%
\thanks{$^{1}$Hung Nguyen is with the Center for Applied Mathematics, Cornell University, Ithaca, NY, USA. (email: \texttt{htn33@cornell.edu}).}%
\thanks{$^{2}$Andreas A. Malikopoulos is with the Applied Mathematics, Systems Engineering, Robotics, Mechanical Engineering, Electrical and Computer Engineering, and Civil \& Environmental Engineering, Cornell University, Ithaca, NY, USA. (email: \texttt{amaliko@cornell.edu}).}}

\begin{document}

\maketitle
\thispagestyle{empty}
\pagestyle{empty}

\begin{abstract}
Sharing of information in multiagent systems allow for improved coordination. Delays in communication can cause asymmetries in information, causing classical results to fail. We study the decentralized LQG problem with delayed information. In general, the optimal policies are not linear and are difficult to calculate. Restricting to linear policies, we analytically calculate the gradient descent for gains from private information. This naturally lends itself to a gradient descent algorithm which we apply to examples to show policy improvements.
\end{abstract}

\section{Introduction}
\label{sec:introduction}
Information structures are defining features which underly decentralized control problems. The importance can be seen when compare results for the Linear-Quadratic-Gaussian (LQG) multiagent control problem under different information structures. Under centralized information, the problem simplifies to that of the single controller which classical results provide a tractable optimal control policy which is linear \cite{wonham1968separation}. For partially nested information, the optimal control policy is still linear but the calculation can be more involved \cite{hochu1972team}. Under delayed information sharing, in general, the optimal control policy is no longer guaranteed to be linear and even the optimal linear controller is unknown \cite{varaiya1978delayed}.

To see how nonclassical information structures, such as delayed information sharing, appear in real world systems, we can first consider the difficulties in the classical information structure. A centralized controller may require information to travel across a large distributed system, which causes concerns over the time it takes to make decisions as well as the security of the information as it travels. In smart power grids, this centralization creates a possible single point of failure in critical infrastructure, which makes makes decentralized architecture more attractive from the viewpoint of privacy, robustness, and security \cite{sun2024trilemma,sousadias2023cybersecurity}.

In this paper, we consider delayed sharing information structure in the decentralized LQG problem. Early works starting in 1974 first studied the problem with a single step of delay \cite{sandell1974nonclassical,kurtaran1974lqg}. Further works analyzed and characterized the results for the one-step delay problem and revealed that moving to the two-step delay problem caused many of these properties to no longer hold \cite{varaiya1978delayed,yoshikawa1978separation}. Following these findings, there is large gap in the literature until the 2010's where the problem sees renewed interest \cite{nayyar2011delayed,nayyar2013common,mahajan2015sufficient,malikopoulos2023}.

The most relevant result is the recent characterization of the optimal linear controller for the finite horizon problem \cite{xu2023delayed}. Though the controller is intractable, the limiting process reveals conditions that good candidates for the optimal linear controller should have.

To find further improvements to such candidate policies, we consider a policy gradient method. Previous works have considered the policy gradient approach for the classical LQR problem \cite{fazel2018} and also the decentralized information structure case \cite{furieri2020firstorder}.

Our main contributions are the following:
\begin{enumerate}
  \item We develop a method to calculate an analytic gradient in the cost with respect to a linear policy update solely on the private information available to team members within our class of admissible policies. This isolates exactly where the the improvements are coming from.
  \item We show that the gradient can be interpretted as a least-squares problem.
  \item We provide examples showing the significant cost reductions obtainable from performing gradient descent using the analytic calculation.
\end{enumerate}

\section{Problem Formulation}
\label{sec:problem}
\subsection{System Model and Information Structure}
We consider a two-member team with the following linear time-invariant system:

\begin{equation}
\begin{aligned}
X_{t+1}
&= AX_t+B^1U_t^1+B^2U_t^2+W_t,\\
Y_t^k
&= C^kX_t+V_t^k,\qquad k\in\{1,2\},\\
\end{aligned}
\label{eq:dynamics}
\end{equation}

\noindent where $X_t \in \mathbb{R}^{d_X}$, $U_t^k \in \mathbb{R}^{d_U^k}$, and $Y_t^k \in \mathbb{R}^{d_Y^k}$. The initial state $X_0$ is drawn a known gaussian distribution. The primitive noises $W_t\sim\mathcal{N}(0,\Sigma_W)$ and $V_t^k\sim\mathcal{N}(0,\Sigma_{V^k})$ have covariance $\Sigma_W \succeq 0$ and $\Sigma_{V_t^k} \succ 0$. The initial state and primitive noises are mutually independent.

For notation purposes, we introduce the following notation to collect quantities and organize.

\begin{equation*}
\begin{gathered}
U_t^M=\operatorname{col}(U_t^1,U_t^2),\quad
Y_t^M=\operatorname{col}(Y_t^1,Y_t^2),\\
B^M=[B^1\ B^2],\quad C^M=\operatorname{col}(C^1,C^2),\\
\Sigma_{V^M}=\operatorname{diag}(\Sigma_{V^1},\Sigma_{V^2}).
\end{gathered}
\end{equation*}

The superscript $M$ allows convenient rewriting of the system, to describe the different team members together as if we were viewing over them as a manager.

Fix the delay in information sharing to be $d \geq 2$. Each team member obtains their individual observations immediately, however their observations and actions are only shared to others after $d$ steps. We describe this difference in information as the common information and private information.

\begin{equation}
\begin{aligned}
\Delta_t
&:= \left(Y_{0:t-d}^k,\ U_{0:t-d}^k\right)_{k\in\{1,2\}},\\
\Lambda_t^k
&:= \left(Y_{t-d+1:t}^k,\ U_{t-d+1:t-1}^k\right),
\qquad k\in\{1,2\}.
\end{aligned}
\label{eq:information}
\end{equation}

The team members are trying to optimized the following shared infinite horizon time-average cost:

\begin{equation}
J(\pi)=\limsup_{T\to\infty}\frac{1}{T}
\sum_{t=0}^{T-1}
\mathbb{E}_{\pi}\!\left[
X_t^\top QX_t+(U_t^M)^\top RU_t^M
\right].
\label{eq:cost}
\end{equation}

\noindent where $Q \succeq 0$ and $R=\operatorname{diag}(R^1,R^2) \succ 0$.

\subsection{Baseline Policy}
Under its solvability condition, the optimal linear controller result in \cite{xu2023delayed} for the finite horizon problem finds that the gain on the common state estimator is determined from matrices satisfying Riccati equations. Assuming that the limiting process produces a convergent steady-state solution, the intractable forward-backward equations that characterize the needed matrices are replaced by steady-state equations.

We will only consider optimizing policies which satisfy these steady-state equations.

\begin{assumption}[Stable Stationary Baseline]
\label{as:stable-stationary-baseline}
Assume that a candidate policy $\pi(0)$ is given such that

\begin{enumerate}
  \item The policy is admissible and stationary.
  \item There are constant matrices $P$ and $K$ such that $P \succeq 0$ and
  \begin{equation}
  \begin{aligned}
  &\mathbb E_{\pi(0)}[U_t^M\mid\Delta_t]=K\widehat X_t,\\
  &\widehat X_t=\mathbb E_{\pi(0)}[X_t\mid\Delta_t],
  \end{aligned}
  \label{eq:common-mean}
  \end{equation}
  and
  \begin{equation}
  \begin{aligned}
  &P=Q+K^\top RK+A_c^\top PA_c,\\
  &0=RK+(B^M)^\top PA_c,\\
  &A_c:=A+B^MK.
  \end{aligned}
  \label{eq:steady-state}
  \end{equation}
  \end{enumerate}

\end{assumption}

Though \eqref{eq:common-mean} is an estimate using policy $\pi(0)$, we will see that \eqref{eq:steady-state} still holds for the estimates under the policy improvement class that we will consider.

\subsection{Policy Improvement Class}
We are looking for linear policy improvements of our initial policy $\pi(0)$.

\begin{equation}
\begin{aligned}
&U_t^k(\Theta)=U_t^k(0)+\Theta^k\eta_t^k,\\
&\eta_t^k:=\operatorname{col}_{a=0}^{d-1}
\left(Y_{t-a}^k-\widehat Y_{t-a\mid t}^k\right).
\end{aligned}
\label{eq:corrected-policy}
\end{equation}

We define $\Theta$ to be the collection of policy changes $\Theta^k$ for each team member, and denote the updated policy as $\pi(\Theta)$. We see that for $\Theta = 0$, we obtain the starting baseline policy $\pi(0)$. Using this notation, we will denote the cost function under the updated policy $\pi(\Theta)$ as,

\begin{equation}
J(\Theta):=J(\pi(\Theta)).
\label{eq:parameterized-cost}
\end{equation}

We will supress the argument $\Theta$ and $\Theta^k$ when the dependence is clear.

\section{Gradient and Least-Squares Formulation}
\label{sec:main}
We want to optimize the gain on private information. We are motivated by the first-order condition

\begin{equation}
\nabla_{\Theta^k}J(\Theta)=0,\qquad k\in\{1,2\}.
\label{eq:stationarity}
\end{equation}

This condition checks that no individual team member can improve the policy by solely changing their own gains on private information. Non-zero gradients will indicate that we can update the policy and descend in cost. To calculate this gradient, we will first transition to delayed error coordinates.

\subsection{Delayed Error Coordinates}
We start by defining the variables needed for the augmented system.

\begin{equation}
\begin{aligned}
\bar X_t
&=\operatorname{col}(X_t,Y_{t-1}^M,\ldots,Y_{t-d+1}^M,
U_{t-1}^M,\ldots,U_{t-d+1}^M),\\
e_t
&=\bar X_t-\widehat{\bar X}_t,\qquad
\widehat{\bar X}_t
=\mathbb E_{\pi(\Theta)}[\bar X_t\mid\Delta_t],\\
\xi_t
&=\operatorname{col}(W_t,V_t^M),\qquad
\Omega=\operatorname{diag}(\Sigma_W,\Sigma_{V^M}).
\end{aligned}
\label{eq:augmented-error}
\end{equation}

We can now write the dynamics for the augmented variable $\bar{X}_{t}$ as,

\begin{equation}
\bar{X}_{t+1}
= \bar{A}\bar{X}_t + \bar{B}U_t^M + \bar{G}\xi_t
\label{eq:augmented-state-dynamics}
\end{equation}

\noindent where $\bar{A}, \bar{B}, \bar{G}$ are appropriate constant matrices.

\begin{lemma}[Update Components]
\label{lem:prediction}
Let $S_x=[I_{d_X}\ \ 0]$ and
\[
E^1=\operatorname{col}(I_{d_U^1},0),\qquad
E^2=\operatorname{col}(0,I_{d_U^2}).
\]
There are constant matrices $H^k,N^k,L_e(0),D_e(0)$, such that
\begin{align}
  \eta_t^k&=H^ke_t+N^k\xi_t,
  \label{eq:eta-realization}\\
  a_t(\Theta)&:=U_t^M-KX_t=L_e(\Theta)e_t+D_e(\Theta)\xi_t,
  \label{eq:a-realization}
\end{align}
where
\begin{equation}
\begin{aligned}
  L_e(\Theta)&=L_e(0)+\sum_k E^k\Theta^kH^k,\\
  D_e(\Theta)&=D_e(0)+\sum_k E^k\Theta^kN^k.
\end{aligned}
\label{eq:prediction-matrices}
\end{equation}
Consequently,
\begin{equation}
  \bar X_{t+1}-\mathbb E[\bar X_{t+1}\mid\Delta_t]
  =T_e(\Theta)e_t+G_e(\Theta)\xi_t,
\label{eq:prediction-error}
\end{equation}
where
\begin{equation}
\begin{aligned}
  T_e(\Theta)&=\bar A+\bar B(KS_x+L_e(\Theta)),\\
  G_e(\Theta)&=\bar G+\bar B D_e(\Theta).
\end{aligned}
\label{eq:TG-matrices}
\end{equation}
\end{lemma}

\begin{proof}
  Let us first start with the residual $\eta_t^k$. See that for the first component,
  \[
    Y_t^k-\mathbb E[Y_t^k\mid\Delta_t]=C^kS_xe_t+V_t^k,
  \]
  so we see that it can be written as a linear combination of $e_t$ and $\xi_t$. For the remaining components, these are the corresponding values as-is from $e_t$.

  For the action residual $a_t$, note that

  \begin{equation*}
  U_t^M=U_t^M(0)+\sum_k E^k\Theta^k\eta_t^k,
  \end{equation*}

  Substituting this and \eqref{eq:common-mean} into \eqref{eq:a-realization}, we find that

  \begin{equation*}
  a_t=L_e(0)e_t+D_e(0)\xi_t+\sum_k E^k\Theta^k\eta_t^k
  \end{equation*}

  Using the previous equation found for $\eta_t^k$, we obtain the second result of \eqref{eq:a-realization}.

  For the remaining result, we have that

  \[
  \begin{aligned}
  &\bar X_{t+1}-\mathbb E[\bar X_{t+1}\mid\Delta_t]\\
  &=\bar A e_t+\bar B(U_t^M-K\widehat X_t)+\bar G\xi_t\\
  &=(\bar A+\bar BKS_x)e_t+\bar B a_t+\bar G\xi_t\\
  &=T_e(\Theta)e_t+G_e(\Theta)\xi_t.
  \end{aligned}
  \]
\end{proof}

\begin{lemma}[Error Dynamics]
  The dynamics of $e_t$ and its covariance is given by

  \begin{equation}
  \begin{aligned}
  &e_{t+1}=A_ee_t+G_e\xi_t,\\
  &\Sigma_e=A_e\Sigma_eA_e^\top+G_e\Omega G_e^\top,\\
  &A_e\Sigma_eS^\top=0
  \end{aligned}
  \label{eq:error-covariance}
  \end{equation}
  where
  \begin{equation}
  \begin{aligned}
  &A_e:=T_e-K_{\mathrm f}S,\\
  &K_{\mathrm f}:=T_e\Sigma_e S^\top\left(S\Sigma_e S^\top\right)^\dagger.
  \end{aligned}
  \label{eq:error-matrices}
  \end{equation}
\end{lemma}

\begin{proof}
  See that \eqref{eq:prediction-error} is almost $e_{t+1}$. The left hand side is only off in the conditioning. Let us quantify the difference.

  \[
  \begin{aligned}
  &\mathbb E_{\pi(\Theta)}[\bar X_{t+1}\mid\Delta_{t+1}]-\mathbb E_{\pi(\Theta)}[\bar X_{t+1}\mid\Delta_t]\\
  &=\mathbb E_{\pi(\Theta)}\big[\bar X_{t+1}-\mathbb E_{\pi(\Theta)}[\bar X_{t+1}\mid\Delta_t]\,\big|\,\Delta_{t+1}\big]\\
  &=T_e\,\mathbb E_{\pi(\Theta)}[e_t\mid\Delta_{t+1}]\\
  &\quad+G_e\,\mathbb E_{\pi(\Theta)}[\xi_t\mid\Delta_{t+1}].
  \end{aligned}
  \]

  The second equality comes from the tower property and the last equality comes from \eqref{eq:prediction-error}. See that the second term is zero since the primitive noise is at time $t$ while $\Delta_{t+1}$ has only revealed additional information about the observation and actions at $t - d + 1$. Allowing $S$ to denote the selector matrix selecting the new common information at $t - d + 1$, we can write
  \[
  S\bar X_t-S\widehat{\bar X}_t=Se_t.
  \]

  Recall that both $e_t$ and $S e_t$ conditioned on $\Delta_t$ have multivariate normal distributions. Thus, conditioning on the realization of the latter, we obtain

  \begin{equation}
  \begin{aligned}
  &\mathbb{E}_{\pi(\Theta)}\!\left[e_t \mid \Delta_{t+1}\right]\\
  &=
  \Sigma_e S^\top
  \left(S\Sigma_e S^\top\right)^\dagger S e_t
  \end{aligned}
  \label{eq:conditional-error-mean}
  \end{equation}

  This gives the following equation for the error dynamics:

  \[
  \begin{aligned}
  &e_{t+1}=T_e e_t+G_e\xi_t\\
  &\quad-\big(T_e\,\mathbb E_{\pi(\Theta)}[e_t\mid\Delta_{t+1}]\big)
  \end{aligned}
  \]
  Substituting again gives us the desired error dynamics equation.

  Since are considering the stationary distribution, we have that

  \begin{equation}
  \begin{aligned}
  \Sigma_e
  &=\operatorname{Cov}(e_{t+1})\\
  &=\operatorname{Cov}(A_e e_t+G_e\xi_t)\\
  &=A_e\Sigma_e A_e^\top+G_e\Omega G_e^\top,
  \end{aligned}
  \label{eq:error-covariance-calculation}
  \end{equation}
  where the resulting cross terms disappear since $\operatorname{Cov}(e_t,\xi_t)=0$ due to timing and primitive noise independence, giving us the third equality.

  To see the final identity, see that
  \[
  \begin{aligned}
  A_e\Sigma_eS^\top
  &=(T_e-K_{\mathrm f}S)\Sigma_eS^\top\\
  &=T_e\Sigma_eS^\top
  -T_e\Sigma_eS^\top(S\Sigma_eS^\top)^\dagger
  (S\Sigma_eS^\top)\\
  &=0.
  \end{aligned}
  \]
\end{proof}

This rewriting reveals the second assumption needed for the results that follow.
\begin{assumption}
\label{as:stable-error}
  Assume that
  \begin{equation}
  \rho(A_e(0))<1.
  \end{equation}
\end{assumption}
This assumption tells us that the effect of previous states, errors, and noise from previous times do not persist as we continue in time. Note that $A_e(\Theta)$ is continuous with respect to $\Theta$, so there exists $\epsilon$ such for $|| \Theta || < \epsilon$, we have that $\rho(A_e(\Theta))<1$. We restrict ourselves to such policy updates $\Theta$, and denote the space of such stable policy updates as $\mathcal{S}$. Additionally note that Assumption~\ref{as:stable-stationary-baseline} still holds in a similar, though with the modified estimator, so there is a unique $P_e\succeq0$ satisfying
\begin{equation}
  P_e=L_e^\top\mathcal RL_e+A_e^\top P_eA_e.
\label{eq:adjoint}
\end{equation}

\subsection{Gradient Calculation}

\begin{lemma}[Stationary cost]
\label{lem:cost}
For $\Theta\in\mathcal S$,
\begin{equation}
\begin{aligned}
 J(\Theta)={}&\operatorname{tr}(P\Sigma_W)
       +\operatorname{tr}(\mathcal R L_e\Sigma_eL_e^\top)\\
       &+\operatorname{tr}(\mathcal R D_e\Omega D_e^\top).
\end{aligned}
\label{eq:cost-trace}
\end{equation}
where
\begin{equation}
\mathcal R=R+(B^M)^\top PB^M.
\label{eq:modified-R}
\end{equation}
\end{lemma}

\begin{proof}
  Note that $U_t^M(\Theta) = K X_t + a_t(\Theta)$. Plugging back into \eqref{eq:dynamics}, we get
  \[
  \begin{aligned}
  X_{t+1}&=AX_t+B^M\left(KX_t+a_t(\Theta)\right)+W_t\\
  &=A_cX_t+B^Ma_t(\Theta)+W_t,
  \end{aligned}
  \]
  This gives the following cost function
  \[
  J(\Theta)=\lim_{T\to\infty}\frac{1}{T}\mathbb E_{\pi(\Theta)}\left[\sum_{0}^{T-1}\ell_t(\Theta)\right]
  \]
  where the stage cost is given by
  \[
  \begin{aligned}
  \ell_t(\Theta):={}&X_t^\top QX_t+X_t^\top K^\top RKX_t\\
  &+2X_t^\top K^\top Ra_t(\Theta)+a_t(\Theta)^\top Ra_t(\Theta).
  \end{aligned}
  \]

  Note that $X_t$ and $a_t(\Theta)$ have bounded second moments under the linear policy $\pi(\Theta)$ and the gaussianity of the system. By stationarity, we have that the second moments become constant. We also have, by stationarity, that $\mathbb E_{\pi(\Theta)}[X_t^\top PX_t]=\mathbb E_{\pi(\Theta)}[X_{t+1}^\top PX_{t+1}]$. This gives us that
  \[
  \begin{aligned}
  &\lim_{T\to\infty}\frac{1}{T}\mathbb E_{\pi(\Theta)}\left[\sum_{t=0}^{T-1}\ell_t(\Theta)\right]\\
  &=\mathbb E_{\pi(\Theta)}[\ell_t(\Theta)]\\
  &=\mathbb E_{\pi(\Theta)}\big[\ell_t(\Theta)+X_{t+1}^\top PX_{t+1}-X_t^\top PX_t\big].
  \end{aligned}
  \]

  Expanding $X_{t+1}$ and simplifying the expression gives
  \[
  J(\Theta)=\tr(P\Sigma_W)+\mathbb E[a_t^\top\mathcal Ra_t].
  \]
  Expanding $a_t(\Theta)$ and simplifying yields the desired result \eqref{eq:cost-trace}.
\end{proof}

\begin{theorem}[Gradient Calculation]
\label{thm:gradient}
For $ \Theta \in \mathcal S$, the cost is differentiable with respect to $\Theta^k$ and the gradient is given by
\begin{equation}
\begin{aligned}
 \nabla_{\Theta^k}J=2\big(&((E^k)^\top\mathcal R L_e
                   +(F^k)^\top P_e A_e)\Sigma_e(H^k)^\top\\
                 +&((E^k)^\top\mathcal R D_e
                   +(F^k)^\top P_e G_e)\Omega(N^k)^\top\big).
\end{aligned}
\label{eq:gradient}
\end{equation}
where
\[
F^k=\bar BE^k.
\]
\end{theorem}
\begin{proof}
  Consider a perturbation $\delta \Theta^k$ in the policy. From \eqref{eq:prediction-matrices}, \eqref{eq:TG-matrices}, and \eqref{eq:error-matrices},
  \[
  \begin{aligned}
  \delta L_e&=E^k(\delta\Theta^k)H^k, & \delta D_e&=E^k(\delta\Theta^k)N^k,\\
  \delta A_e&=F^k(\delta\Theta^k)H^k-(\delta K_{\mathrm f})S, & \delta G_e&=F^k(\delta\Theta^k)N^k.
  \end{aligned}
  \]
  From \eqref{eq:error-covariance},
  \[
  \begin{aligned}
  \delta\Sigma_e={}&A_e(\delta\Sigma_e)A_e^\top\\
  &+(\delta T_e)\Sigma_eA_e^\top+A_e\Sigma_e(\delta T_e)^\top\\
  &-(\delta K_{\mathrm f})S\Sigma_eA_e^\top-A_e\Sigma_eS^\top(\delta K_{\mathrm f})^\top\\
  &+(\delta G_e)\Omega G_e^\top+G_e\Omega(\delta G_e)^\top.
  \end{aligned}
  \]
  But note that from \eqref{eq:error-covariance}, we have
  \[
  \begin{aligned}
  \delta\Sigma_e={}&A_e(\delta\Sigma_e)A_e^\top+(\delta T_e)\Sigma_eA_e^\top+A_e\Sigma_e(\delta T_e)^\top\\
  &+(\delta G_e)\Omega G_e^\top+G_e\Omega(\delta G_e)^\top.
  \end{aligned}
  \]

  Differentiating the terms in \eqref{eq:cost-trace}, we get
  \[
  \begin{aligned}
  \delta\tr(P\Sigma_W)&=0,\\[3pt]
  \delta\tr(\mathcal RL_e\Sigma_eL_e^\top)&=2\tr\big(\mathcal RL_e\Sigma_e(\delta L_e)^\top\big)\\
  &\quad+\tr\big(L_e^\top\mathcal RL_e\,\delta\Sigma_e\big)\\
  &=2\tr\big(\mathcal RL_e\Sigma_e(\delta L_e)^\top\big)\\
  &\quad+2\tr\big(P_eA_e\Sigma_e(\delta T_e)^\top\big)\\
  &\quad+2\tr\big(P_eG_e\Omega(\delta G_e)^\top\big),\\[3pt]
  \delta\tr(\mathcal RD_e\Omega D_e^\top)&=2\tr\big(\mathcal RD_e\Omega(\delta D_e)^\top\big).
  \end{aligned}
  \]
  The use of \eqref{eq:adjoint} giving the equality for the second differential term. Adding up the terms gives us that
  \[
  \begin{aligned}
  \delta J=2\Big(&\tr\big(\mathcal RL_e\Sigma_e(\delta L_e)^\top\big)+\tr\big(\mathcal RD_e\Omega(\delta D_e)^\top\big)\\
  &+\tr\big(P_eA_e\Sigma_e(\delta T_e)^\top\big)+\tr\big(P_eG_e\Omega(\delta G_e)^\top\big)\Big).
  \end{aligned}
  \]

  Substituting the perturbations found above, we can rearrange and factor out $\delta \Theta^k$. Using the identity that
  \[
  \delta J=\tr\Big(\big(\nabla_{\Theta^k}J(\Theta)\big)^\top\delta\Theta^k\Big)
  \]
  and by identifying the corresponding factor, we have shown \eqref{eq:gradient}.
\end{proof}

We now have the explicit calculation of the gradient of $J(\Theta)$ with respect to private information linear policy update of team member $k$. Before we continue to the numerical results, we will first consider this problem from the least squares perspective to see how to interpret the direction of improvement but as well provide an alternative method of implementation.

\subsection{Least-Squares Formulation}
For the policy update $\Theta$, define the following quantities:
\begin{equation}
  r_t^k := (E^k)^\top \mathcal{R}\,a_t(\Theta)
    + (F^k)^\top P_e e_{t+1},
  \label{eq:cost-residual}
\end{equation}
\begin{equation}
  \Sigma_{\eta^k}
  := \mathbb{E}_{\pi(\Theta)}\!\left[
       \eta_t^k (\eta_t^k)^\top \right].
  \label{eq:innovation-covariance}
\end{equation}

\begin{corollary}
\label{cor:least-squares}
Suppose $\Sigma_{\eta^k} \succ 0$. Then
\begin{equation}
\begin{aligned}
  \mathcal C^k(\Theta)&:=\operatorname*{arg\,min}_{\mathcal C}\mathbb E_{\pi(\Theta)}\big\|r_t^k-\mathcal C\eta_t^k\big\|_2^2\\
  &=\frac{1}{2}\nabla_{\Theta^k}J(\Theta)\,\Sigma_{\eta^k}^{-1}.
\end{aligned}
  \label{eq:least-squares}
\end{equation}
\end{corollary}

\begin{proof}
  First note that since $\Sigma_{\eta^k} \succ 0$, the least-squares problem has a unique minimizer and the corresponding normal equation is given by
  \[
  \mathcal C\Sigma_{\eta^k}=\mathbb E_{\pi(\Theta)}\big[r_t^k(\eta_t^k)^\top\big].
  \]
  Expanding, we get that
  \[
  \begin{aligned}
  \mathbb E[r_t^k(\eta_t^k)^\top]={}&(E^k)^\top\mathcal R\,\mathbb E[a_t(\eta_t^k)^\top]\\
  &+(F^k)^\top P_e\,\mathbb E[e_{t+1}(\eta_t^k)^\top].
  \end{aligned}
  \]
  Expanding the first term provides
  \[
  \begin{aligned}
  \mathbb E[a_t(\eta_t^k)^\top]={}&L_e\mathbb E[e_te_t^\top](H^k)^\top\\
  &+L_e\mathbb E[e_t\xi_t^\top](N^k)^\top\\
  &+D_e\mathbb E[\xi_te_t^\top](H^k)^\top\\
  &+D_e\mathbb E[\xi_t\xi_t^\top](N^k)^\top
  \end{aligned}
  \]
  and the second term provides
  \[
  \begin{aligned}
  \mathbb E[e_{t+1}(\eta_t^k)^\top]={}&A_e\mathbb E[e_te_t^\top](H^k)^\top\\
  &+A_e\mathbb E[e_t\xi_t^\top](N^k)^\top\\
  &+G_e\mathbb E[\xi_te_t^\top](H^k)^\top\\
  &+G_e\mathbb E[\xi_t\xi_t^\top](N^k)^\top.
  \end{aligned}
  \]

  Note that $e_t$ depends only on primitive noise up to time $t - 1$ but $\xi_t$ only contains the primitive noise at time $t$, so $\mathbb E[e_t\xi_t^\top]=0$. Simplifying,
  \[
  \begin{aligned}
  &\mathbb E[a_t(\eta_t^k)^\top]=L_e\Sigma_e(H^k)^\top+D_e\Omega(N^k)^\top,\\
  &\mathbb E[e_{t+1}(\eta_t^k)^\top]=A_e\Sigma_e(H^k)^\top+G_e\Omega(N^k)^\top.
  \end{aligned}
  \]

  Thus, we see that
  \begin{equation}
  \mathcal C\Sigma_{\eta^k}=\mathbb E_{\pi(\Theta)}\big[r_t^k(\eta_t^k)^\top\big]=\frac{1}{2}\nabla_{\Theta^k}J(\Theta).
  \label{eq:normal-equation}
  \end{equation}
  and rearranging provides \eqref{eq:least-squares}
\end{proof}

The quantity $r_t^k$ signals in some sense what alternative action direction would have been preferable, where the first term corresponds to consideration for the immediate cost and the second term contains $P_e$ which corresponds to consideration for future costs. Due the the policy improvement class, we can only try to best choose $\mathcal{C}$ to align $\mathcal{C} \eta_t^k$.

\section{Numerical Results}
\label{sec:numerical}

\subsection{Implementation}
Using the least-squares interpretation, we can calculate the gradient direction. Since we know that we are stable within a small neighborhood around our current policy, we can scale our update. See Algorithm 1.

\begin{algorithm}[t]
\caption{Stationary policy improvement}
\label{alg:improvement}
\small
\begin{algorithmic}[1]
\Require Baseline policy satisfying
    Assumption~\ref{as:stable-stationary-baseline} and
    Assumption~\ref{as:stable-error},
    starting step size $\overline\alpha>0$,
    minimum step size $\alpha_{\min}>0$,
    step size reduction factor $\beta\in(0,1)$,
    minimum improvement threshold factor $c\in(0,1)$,
    minimum gradient size $\varepsilon>0$
\State $\Theta\gets0$
\Loop
    \State Form necessary matrices from \eqref{eq:prediction-matrices},
        \eqref{eq:TG-matrices}, \eqref{eq:error-covariance},
        \eqref{eq:error-matrices}, and \eqref{eq:adjoint}
    \State Calculate policy cost $J(\Theta)$ using \eqref{eq:cost-trace}
    \State Compute $\nabla J(\Theta)$ using \eqref{eq:gradient}
    \If{$\|\nabla J(\Theta)\|_F\leq\varepsilon$}
        \State \Return $\Theta$
    \EndIf
    \For{$k\in\{1,2\}$}
        \State Compute $\Sigma_{\eta^k}$ using
            \eqref{eq:eta-realization} and
            \eqref{eq:innovation-covariance}
        \State Solve \eqref{eq:normal-equation}
            for $\mathcal C^k$
    \EndFor
    \State $p\gets-\{\mathcal C^1,\mathcal C^2\}$, $\alpha\gets\overline\alpha$
    \While{$\Theta+\alpha p\notin\mathcal S$, or fails
        $J(\Theta+\alpha p)
        \leq J(\Theta)
        +c\alpha\langle\nabla J(\Theta),p\rangle$}
        \State $\alpha\gets\beta\alpha$
        \If{$\alpha<\alpha_{\min}$}
            \State \Return $\Theta$
        \EndIf
    \EndWhile
    \State $\Theta\gets\Theta+\alpha p$
\EndLoop
\end{algorithmic}
\end{algorithm}

\subsection{Example 1}
Consider the following system:
\[
A=\begin{bmatrix}
1.15 & 0.25\\
0.25 & 1.15
\end{bmatrix},
\qquad
B^1=\begin{bmatrix}0.6\\1\end{bmatrix},
\qquad
B^2=\begin{bmatrix}1\\0.6\end{bmatrix},
\]
\[
C^1=\begin{bmatrix}1&0\end{bmatrix},
\qquad
C^2=\begin{bmatrix}0&1\end{bmatrix},
\]
\[
\Sigma_W=4I_2,
\qquad
\Sigma_{V^M}=0.5I_2,
\]
\[
Q=I_2,
\qquad
R=0.04I_2,
\qquad
d=5.
\]

To construct the baseline policy $\pi(0)$, we use the method outlined in \cite{xu2023delayed}. We confirm that Assumption~\ref{as:stable-stationary-baseline} holds with

\[
P \approx
\begin{bmatrix}
1.0984673400 & -0.0682999027\\
-0.0682999027 & 1.0984673400
\end{bmatrix},
\]
\[
K \approx
\begin{bmatrix}
0.4955212920 & -1.3574480716\\
-1.3574480716 & 0.4955212920
\end{bmatrix}.
\]
\[
A_c=A+B^M K \approx
\begin{bmatrix}
0.0898647036 & -0.0689475510\\
-0.0689475510 & 0.0898647036
\end{bmatrix}.
\]
\[
B^M=\begin{bmatrix}0.6&1\\1&0.6\end{bmatrix},
\qquad Q=I_2,
\qquad R=0.04I_2.
\]
\[
\rho(A_c)\approx0.1588122546<1.
\]

We confirm that Assumption~\ref{as:stable-error} holds with
\[
\rho(A_e(0))=0.413740.
\]

From Fig.~\ref{fig:example1}, we see a $28.2348\%$ reduction from the baseline cost.

\begin{figure}[t]
  \centering
  \includegraphics[width=\columnwidth]{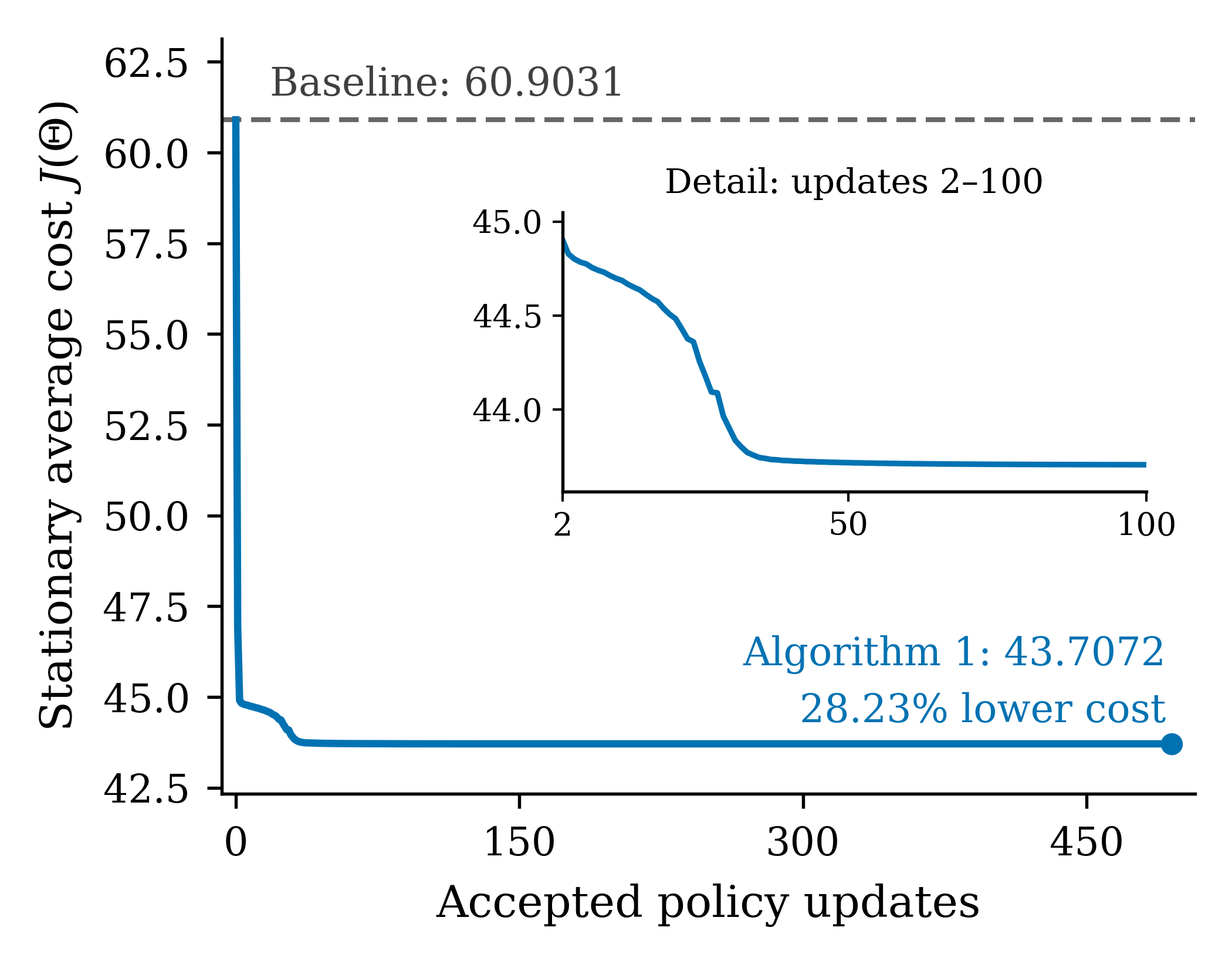}
  \caption{}
  \label{fig:example1}
\end{figure}

\subsection{Example 2}
Consider the following system:
\[
A=\begin{bmatrix}
1.08 & 0.18 & 0.18\\
0.18 & 1.08 & 0.18\\
0.18 & 0.18 & 1.08
\end{bmatrix},
\]
\[
B^1=\begin{bmatrix}0.6\\1\\1\end{bmatrix},
\qquad
B^2=\begin{bmatrix}1\\0.6\\1\end{bmatrix},
\qquad
B^3=\begin{bmatrix}1\\1\\0.6\end{bmatrix},
\]
\[
C^1=\begin{bmatrix}1&0&0\end{bmatrix},
\quad
C^2=\begin{bmatrix}0&1&0\end{bmatrix},
\quad
C^3=\begin{bmatrix}0&0&1\end{bmatrix},
\]
\[
\Sigma_W=4I_3,
\qquad
\Sigma_{V^M}=0.5I_3,
\]
\[
Q=I_3,
\qquad
R=0.04I_3,
\qquad
d=5.
\]

We have that Assumption~\ref{as:stable-stationary-baseline} holds with
{%
\setlength{\arraycolsep}{2pt}
\[
P \approx
\begin{bmatrix}
1.1152443324 & -0.0515229103 & -0.0515229103\\
-0.0515229103 & 1.1152443324 & -0.0515229103\\
-0.0515229103 & -0.0515229103 & 1.1152443324
\end{bmatrix},
\]
\[
K \approx
\begin{bmatrix}
1.0517704858 & -0.8011988778 & -0.8011988778\\
-0.8011988778 & 1.0517704858 & -0.8011988778\\
-0.8011988778 & -0.8011988778 & 1.0517704858
\end{bmatrix},
\]
\[
\begin{aligned}
&A_c=A+B^M K\\
&\approx
\begin{bmatrix}
0.1086645358 & -0.0501477188 & -0.0501477188\\
-0.0501477188 & 0.1086645358 & -0.0501477188\\
-0.0501477188 & -0.0501477188 & 0.1086645358
\end{bmatrix}.
\end{aligned}
\]
}
\[
B^M=
\begin{bmatrix}
0.6&1&1\\
1&0.6&1\\
1&1&0.6
\end{bmatrix},
\qquad
Q=I_3,
\qquad
R=0.04I_3.
\]
\[
\lambda_{\min}(P)\approx1.01220>0,
\qquad
\rho(A_c)\approx0.158812<1.
\]

We have that Assumption~\ref{as:stable-error} holds with
\[
\rho(A_e(0))\approx0.320423.
\]

This example highlights that, though derived for two-member teams, the method is applicable to larger teams. See Figure 2.

\begin{figure}[t]
  \centering
  \includegraphics[width=\columnwidth]{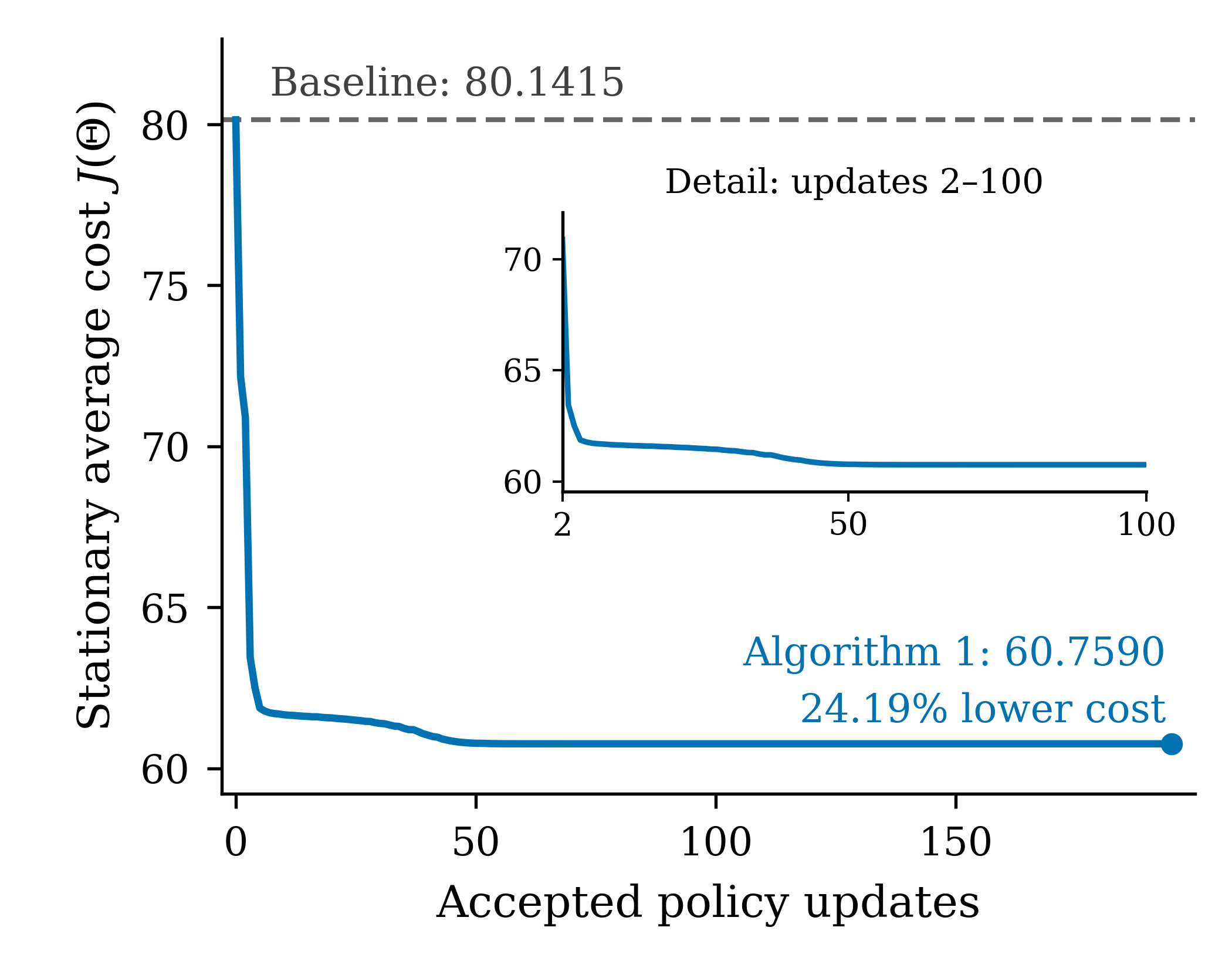}
  \caption{}
  \label{fig:example2}
\end{figure}

\section{Conclusion}
\label{sec:conclusion}
We have found that given appropriate candidate policies, the gradient of the cost function with respect to individual team members' gain on private information can be analyically computed. This method provides a way to check for policy improvements from information.

This method is limited in the fact that we cannot find policy improvements that rely on policy changes among multiple members, and requires assumptions on the baseline policy. Future works looks towards increasing the policy improvement class and alternative assumptions on the baseline policy.

\section*{Acknowledgement of Language Model Usage}

OpenAI Codex \cite{openai2025codex} and Claude Code \cite{anthropic2025claudecode} were used in the literature review, coding, and editing process.

\end{document}